\documentclass[amsmath,amssymb,pra,aps,mathbbm,superscriptaddress,twocolumn]{revtex4}
\usepackage{amsmath,amsfonts,amssymb,amsthm,graphics,graphicx,epsfig,bbm}
\usepackage[colorlinks=true,citecolor=blue,linkcolor=blue,urlcolor=blue]{hyperref}
\usepackage[usenames]{color}
\usepackage{graphicx}
\usepackage{subfigure}
\usepackage{epsfig}
\usepackage{dcolumn}
\usepackage{bm}
\usepackage{color}
\usepackage{verbatim}
\usepackage{epstopdf}
\usepackage{amssymb} 
\usepackage{amstext}
\usepackage{latexsym}
\usepackage{hyperref}
\usepackage{amsfonts}
\usepackage{psfrag}
\usepackage{color}
\usepackage{times}

\newtheorem{thm}{Theorem}
\newtheorem{defn}{Definition}%[section]
\newtheorem*{thm*}{Theorem}
\newtheorem*{defn*}{Definition}
\newtheorem{lem}{Lemma}
\newtheorem{cor}{Corollary}

\begin{document}
\title{Entanglement classification with matrix product states}
\author{M. Sanz\footnote{Corresponding author: mikel.sanz@ehu.eus}}
\affiliation{Department of Physical Chemistry, University of the Basque Country UPV/EHU, Apartado 644, 48080 Bilbao, Spain}
\author{I. L. Egusquiza}
\affiliation{Department of Theoretical Physics and History of Science, University of the Basque Country UPV/EHU, Apartado 644, 48080 Bilbao, Spain}
\author{R. Di Candia}
\affiliation{Department of Physical Chemistry, University of the Basque Country UPV/EHU, Apartado 644, 48080 Bilbao, Spain}
\author{H. Saberi}
\affiliation{Department of Optics, Faculty of Science, Palack\'{y} University, 17. listopadu 12, 77146 Olomouc, Czech Republic}
\affiliation{Department of Physics and Center for Optoelectronics and Photonics Paderborn (CeOPP), University of Paderborn, Warburger Stra\ss e 100, 33098 Paderborn, Germany}
\author{L. Lamata}
\affiliation{Department of Physical Chemistry, University of the Basque Country UPV/EHU, Apartado 644, 48080 Bilbao, Spain}
\author{E. Solano}
\affiliation{Department of Physical Chemistry, University of the Basque Country UPV/EHU, Apartado 644, 48080 Bilbao, Spain}
\affiliation{IKERBASQUE, Basque Foundation for Science, Maria Diaz de Haro 3, 48013 Bilbao, Spain}

\date{\today}

\begin{abstract}
We propose an entanglement classification for symmetric quantum states based on their diagonal matrix-product-state (MPS) representation. The proposed classification, which preserves the stochastic local operation assisted with classical communication (SLOCC) criterion, relates entanglement families to the interaction length of Hamiltonians. In this manner, we establish a connection between entanglement classification and condensed matter models from a quantum information perspective. Moreover, we introduce a scalable nesting property for the proposed entanglement classification, in which the families for $N$ parties carry over to the $N+1$ case. Finally, using techniques from algebraic geometry, we prove that the minimal nontrivial interaction length $n$ for any symmetric state is bounded by $n \leq \lfloor N / 2 \rfloor + 1$.
\end{abstract}


\maketitle
 Entanglement is widely considered the cornerstone of quantum information and an essential resource for relevant quantum effects, such as quantum teleportation~\cite{BBCJPW93,BPMEWZ97,RHRHRBLKBSKJB04,BCSBIJKLLOW04}, quantum cryptography~\cite{BB84,GRTZ02}, or the speed-up of quantum computing~\cite{NC00}, as in Shor's algorithm~\cite{Sh97}. Moreover, entanglement is recognized as useful for understanding properties of condensed matter models, such as quantum phases~\cite{OAFF02} and topological orders~\cite{JWB12}, among others. Entanglement based properties are usually challenging to study both experimentally and theoretically, due to the exponential growth of the associated quantum degrees of freedom. Experimentally, due to the exponentially large amount of degrees of freedom typically involved, the detection and the quantification of the entanglement are difficult to achieve. Theoretically, quantities describing the entanglement are generally complicated function of the quantum state, and they are normally arduous to analyse. To overcome these obstacles, advanced quantum information techniques have been successfully applied to answer condensed matter questions~\cite{ECP10, CPW11, EFG15}, shedding a distinct light on the problem. This novel approach may bring about exciting results in many-body physics, and result in a new revolution in physics, where quantum information and matter phenomena can be formally unified~\cite{ZCZW15}.  
 
An important question in quantum information is the classification of entanglement by means of some mathematical or physical equivalence. Classifying entanglement should help in recognizing similarity between different entangled states, and it may be useful to boost the practicabilities of quantum information protocols. A first result is that quantum states connected by SLOCC operations, which perform probabilistically the same quantum tasks, can be collected into entanglement classes, called SLOCC classes, but also known as SLOCC criterion~\cite{DVC00}. Nevertheless, there is an infinite number of SLOCC classes for four or more parties that may be gathered, in turn, into a finite number of entanglement families~\cite{VDdMV02,LLSS06,LLSS07,OS05,BKMGLS09,WDGC13}. Unfortunately, the community has not been able to relate all classes and families to specific properties or quantum information tasks, although a few of them have certainly raised experimental interest~\cite{photonsPan,ionsInnsbruck,cQEDYale,photonsWeinfurter}. It is noteworthy to mention that, up to now,  no general characterization nor classification of entanglement exist for many-body systems. 
 
In this Article, we present an entanglement classification for quantum states induced by their MPS structure, which preserves the SLOCC criterion and is exemplified for the symmetric subspace.  The proposed classification is based on the local properties of the multipartite quantum state. In this sense, it relates entanglement families to the interaction length of Hamiltonians, establishing a connection between entanglement classification and condensed matter. Our proposal is twofold beneficial: on the one hand, it does not result in an infinite number of entanglement classes, considerably simplifying their study; on the other hand, it provides a direct  physical insight to the nonlocality of entanglement classes, given by the interaction length of their parent Hamiltonians.  Additionally, we introduce a scalable nesting property in which the families for $N$ parties carry over to the $N+1$ case.

We focus on the classification of the SLOCC classes corresponding to symmetric states, which are invariant under any permutation of the parties~\cite{TG09}, i.e. \( F |\psi\rangle= |\psi\rangle \) where \( F \) is an exchange operator. This is an interesting subspace, since its dimension grows linearly with the number of parties but, at the same time, it contains many physically relevant states.

A pure state $|\Psi\rangle$ is called \emph{entangled} if it is not \emph{separable}~\cite{HHHH09}, i.e. if it cannot be written as a tensor product $|\Psi\rangle \neq |\Psi_1 \rangle \otimes |\Psi_2\rangle \otimes \cdots \otimes |\Psi_N \rangle$. This definition is indirect and this may be the reason that its quantification for more than two parties is still unsettled. The idea of separability emerges when one tries to identify which states can be generated from other states, defined locally in each subsystem, by using local operations and classical communication among parties. Therefore, these local operations provide a natural criterion to collect quantum states in \emph{entanglement classes} with the same type of entanglement \cite{DVC00}. More specifically, two states belong to the same class if they can be transformed into each other with non-zero probability via SLOCC operations.

The splitting of the Hilbert space into the SLOCC classes is fully understood for bipartite and tripartite qubit systems, even in the nonsymmetric case \cite{DVC00}. There is only one entanglement class for two qubits: the one containing the Einstein-Podolski-Rosen (EPR) state $|\Psi\rangle = \frac{1}{\sqrt{2}} \left ( |00\rangle + |11\rangle \right )$, and two symmetric classes for three qubits:   one represented by the Greenberger--Horne--Zeilinger (GHZ) state $|\Psi \rangle = \frac{1}{\sqrt{2}} \left ( |000\rangle +|111\rangle \right)$, and the second one referenced to the $W$ state $|\Psi\rangle = \frac{1}{\sqrt{3}} \left( |100\rangle + |010 \rangle + |001\rangle \right )$. However, for four or more qubits, the number of SLOCC classes is infinite, and their parametrization grows exponentially with the number of parties, while lacking robustness against experimental errors~\cite{DVC00}. In this sense, SLOCC classification makes the association of classes to specific physical properties difficult. This explains, so far, the lack of experimental interest in producing states beyond the well-known GHZ or $W$, among few others.

\begin{figure*}[t] 
  \centering
  \includegraphics[angle=0,width=\textwidth]{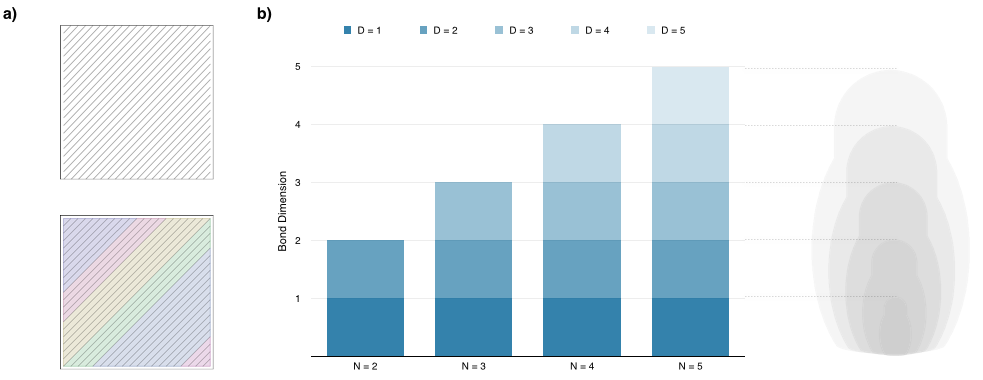}
  \caption{\label{fig}
   \textbf{a}, The SLOCC criterion divides the Hilbert space (the square) in such a way that every quantum state is in a well defined class (the lines). For four or more parties, the number of these SLOCC classes is infinite. However, they may be gathered into families (the colored areas) under certain rules, ideally with more physical associations than mathematical ones. Here, the condition is given by the minimal bond dimension of the matrix-product-state representation of quantum states, relating the MPS classes to the interaction length of parent Hamiltonians. \textbf{b}, The proposed MPS classification enjoys a scalable nesting property in which the classes of an $N$-partite family can be mapped onto the classes of the corresponding $(N+1)$ case, generating a {\it matryoushka structure}. A detailed example is given for the symmetric subspace of arbitrary number of parties.}
\end{figure*}

Due to the natural interest of the SLOCC criterion, it is customary to collect these infinite SLOCC classes into a finite number of larger families based on specific mathematical properties \cite{VDdMV02,LLSS06,LLSS07,WDGC13} or operational approaches \cite{BKMGLS09,BTvZLSA09} (see Fig.~\ref{fig}a). However, up to now, all classifications have failed to associate groups of classes or families to specific quantum information tasks. Since there are infinite SLOCC classes for four or more parties, they can be partitioned into families in an infinite number of ways, and we would not expect all to be relevant. To overcome this conundrum, we propose the following criteria that an SLOCC classification into families should fulfill: 
1) Every SLOCC class must belong to only one family (classes should not cross any border separating families), 2) separable states must be in one family and this family should contain only separable states, 3) SLOCC classes belonging to the same family must show common physical and/or mathematical properties, and 4) the classification into families must be efficient in the sense that (a) the number of families must grow ``slowly" with the number of qubits, and (b) the efforts for classifying $N$ qubits should be useful for classifying $N+1$~(nesting).

In the last decades, matrix product states (MPS) and tensor networks have emerged as a powerful tool to tackle complex problems in many-body systems \cite{PGVWC07,SWPGC09,CSWCPG13,Mi11}. In this sense, any quantum state $|\Psi \rangle = \sum_{i_1, \ldots , i_N =1}^d a_{i_1 \ldots i_N} |i_1 \ldots i_N \rangle$ can be rewritten in a local manner as  $|\Psi \rangle = \sum_{i_1, \ldots , i_N =1}^d A^{[1]}_{i_1} \cdots A^{[N]}_{i_N} |i_1 \ldots i_N \rangle$, where $A^{[k]}_{j_k}$ are matrices containing all local information related to site $k$. This language is convenient to describe ground states of local Hamiltonians \cite{SWPGC09,Mi11,HHS10}, sequential processes \cite{SeqMPS05,SeqOper08,Hamed08}, and systems fulfilling an area law \cite{Ha07}.
In the following, we use the MPS formalism to define an entanglement classification of quantum states into families of SLOCC classes, based on the local dimension of the matrices describing the states, called \emph{bond dimension} (see Fig.~\ref{fig}b). Furthermore, we prove that this MPS classification is directly related to the interaction length of the corresponding parent Hamiltonian \cite{Mi11}. We apply these novel concepts to the case of the symmetric subspace, although they could be extended to more restricted or more general sets of states. We start by highlighting that any symmetric state admits MPS representations with site-independent diagonal matrices. If we further request a minimal bond dimension, we must find the optimal way to represent any symmetric state as $|\Psi^{(N)}\rangle = \sum_{i=1}^D |v_i \rangle^{\otimes N}$. The dimension of the matrix parallels the number of  the Schmidt coefficients, which  has been proposed to quantify entanglement \cite{EB01}.  Hereafter, when we refer to the bond dimension of a symmetric state, we mean the {\it minimal} bond dimension associated with its diagonal representation.

Crucially, SLOCC transformations preserve the minimal bond dimension (see Supplementary Information). Indeed, if $|\Psi_A^{(N)}\rangle$ and $|\Psi_B^{(N)}\rangle$ are two quantum states with bond dimensions $D_A$ and $D_B$ respectively, and $\mathcal{C}$ a class, then
\begin{equation}
|\Psi_A^{(N)}\rangle, \, |\Psi_B^{(N)}\rangle \in \mathcal{C} \Rightarrow D_A = D_B .
\end{equation}
This implies that all states of the same SLOCC class may be represented with the same minimal matrix dimension, which is a SLOCC invariant. In this way, we can define a family of SLOCC classes by means of the following equivalence relation:

\begin{defn*}[Diagonal MPS entanglement classification] Let $S_A$ and $S_B$ be SLOCC classes, and $D_A$, $D_B$ the minimal bond dimension of their respective states in its diagonal MPS representation. We say that $S_A \sim S_B \Leftrightarrow D_A=D_B$, and we call entanglement families the resulting equivalence classes.
\end{defn*}

Notice that the class of separable states can be optimally represented with matrices with bond dimension $D=1$, and, indeed, it coincides with the family $D=1$. Therefore, the diagonal MPS (DMPS) classification, proposed here for the symmetric subspace, fulfills the first two aforementioned conditions. 
Moreover, in this DMPS classification any symmetric state of $N$ qubits can be written with at most bond dimension $N$, so the number of families grows linearly with the number of parties involved. In this sense, the DMPS classification also satisfies a recently proposed tractability criteria~\cite{WDGC13}.

The explicit translational invariance of the MPS formulation leads the DMPS classification to fulfill the aforementioned criterion (4b). Let $\{A_i \}_{i}$ define an $N$-partite symmetric state $|\Psi^{(N)} \rangle = \sum_{i's} {\rm tr} \left ( A_{i_1} \cdots A_{i_N} \right ) |i_1 \ldots i_N \rangle$. Then, the state $|\Psi^{(N+1)} \rangle = \sum_{i's} {\rm tr} \left ( A_{i_1} \cdots A_{i_N} A_{i_{N+1}} \right ) |i_1 \ldots i_N i_{N+1}\rangle$, which lives in a different Hilbert space, namely that of \( N+1 \) parties, does show exactly the same local properties as \( |\Psi^{(N)} \rangle \). As the DMPS classification respects these local properties, and not just for a given state but for the whole SLOCC class, an intriguing~\emph{nesting property} of the families for different number of parties emerges:
\begin{thm*}[Nesting]
Let us consider an $N$-particle symmetric state of qubits $|\psi_N\rangle=\sum_{k=1}^{D(\psi_N)} |x_k\rangle^{\otimes N}$ with optimal bond dimension $D(\psi_N)$, such that $D(\psi_N)\leq \lfloor N/2\rfloor +1$. Then, the state $|\psi_{N+1}\rangle=\sum_{k=1}^{D(\psi_N)}|x_k\rangle^{\otimes (N+1)}$ has optimal bond dimension $D(\psi_{N+1})=D(\psi_N)$.
\end{thm*}
See Supplementary Information for the proof. This theorem shows that, from the perspective of the local properties, the only purely $N$-partite states are the ones whose optimal bond dimension is larger than the maximal bond dimension of any state with $N-1$ parties. This generates a \emph{matryoushka structure} depicted in Fig.~\ref{fig}b, where, unlike other entanglement classifications, the classification for $N$ parties is connected with the classification for $N+1$, for all $N$. We believe that a further exploitation of this scalable nesting property would be interesting in the many-body case, where the exact number of particles is usually not relevant. 

Lastly, in the MPS formalism, the role of parent Hamiltonians for a given state comes to the fore. Namely, for any given state, one can construct a local Hamiltonian which includes it in its ground eigenspace. The MPS representation informs us about the features and interaction length of this construction.  For instance, it controls whether it is the only ground state or there is a spontaneous symmetry breaking \cite{PGVWC07}, or the inheritance of local and global symmetries \cite{SWPGC09,Mi11}. This constructive method works as follows: Let $|\Psi^{(N)}\rangle$ be a symmetric pure state of $N$ parties with an MPS representation with bond dimension~$D$. The reduced density matrices for $n \leq N$ are defined by $\rho^{(n)} = {\rm tr}_{N-n} \left ( |\Psi^{(N)}\rangle \langle \Psi^{(N)}| \right )$, i.e. by tracing out \( N-n \) parties. By construction, ${\rm rank} (\rho^{(n)}) \leq D$, so when $n >  \log_2 D$, $\rho^{(n)}$ has a kernel. Let $h^{(n)} \geq 0$ be the projector onto this kernel and $H = \sum_{i=1}^N h_i^{(n)}$ the total Hamiltonian. Thus, $H$ is a positive operator and $|\Psi^{(N)}\rangle$ is in its ground manifold since $H|\Psi^{(N)}\rangle = 0$. Additionally, the interaction length of $H$ is $n$. 
In order to apply this construction to the DMPS classification, notice that all reduced density matrices for more than one party have nontrivial kernel when acting on the full \( n \)-qubit Hilbert space, since their support is restricted to symmetric states. We must then consider the relevant kernel in the corresponding symmetric subspace of dimension \( n+1 \). Clearly, if the bond dimension is \( D \), the rank of \( \rho^{(n)} \) is at most \( D \) in the symmetric space as well. Thus, we have ensured the existence of a relevant parent Hamiltonian with interaction length \( n \), if \( n\geq D \). However, techniques from algebraic geometry allow us to prove that the minimal nontrivial interaction length for any symmetric state is bounded as $n\leq\lfloor{N/2}\rfloor+1$ (see Supplementary Information). We shall see from the examples below, indeed, that interaction lengths shorter than $D$ do arise. Notice that in our construction we propose families of parent Hamiltonians, for all of which all the states in an entanglement family are ground states. In particular, the ground manifold is always degenerate in our construction. To put this statement in context, bear in mind that the concept of parent Hamiltonian of a state is a general one: a local Hamiltonian which includes the state in its ground manifold . A particularly well established technique for constructing parent Hamiltonians for states or families of states is given by the MPS construction. In that case, the uniqueness of the ground state of this MPS Hamiltonian will depend on certain properties of the corresponding matrix. In particular on the property (or otherwise) of injectivity, such that a lower bound for the interaction length is given by a polynomial in $D$ and $d$ \cite{wielandt}. A particular example of this for the GHZ state may be found in Ref.~\cite{facchi}

{\it i)} {\it $GHZ$ states.---}  The GHZ states may be represented, independently of \(N \), by the matrices \( A_0= \mathrm{diag}\{1,0\} \) and \( A_1= \mathrm{diag}\{0,1\} \). The corresponding parent Hamiltonian family reads
\begin{equation*}
H_{\rm GHZ}=J \sum_{i=1}^N \boldsymbol{ \sigma}_i\cdot \boldsymbol{ \sigma}_{i+1}- J_z\sum_{i=1}^N \sigma_i^z \sigma_{i+1}^z\,,
\end{equation*}
with the conditions \( J_z>0 \) and \( J_z>2J \). 

{\it ii)} {\it $W$ states.---} The \( W_N \) states generalize to \( N \) parties the structure of the three-partite \( W \) state, and they correspond to Dicke states with just one excitation. They can be represented with bond dimension \( D(W_N)=N \) by the following sequence of diagonal matrices \( A_0= \lambda\, \mathrm{diag}\{0,1,e^{i\varphi},\ldots,e^{ i k\varphi},\ldots, e^{ i(N-2)\varphi}\}\), with \( \varphi=2\pi/N(N-1) \), and \( A_1= \mu \,\mathrm{diag}\{(1-N)^{1/N},1,e^{i(N-1)\alpha},\ldots, e^{i (N-k) \alpha},\ldots, e^{2i \alpha}\} \), where \( \alpha=2\pi/N \). Then, we propose the parent Hamiltonian family   
\begin{equation*} 
H_W= J \sum_{i=1}^N \left[-2 \sigma_i^z+ \sigma_i^z \sigma_{i+1}^z- \gamma \boldsymbol{ \sigma}_i\cdot \boldsymbol{ \sigma}_{i+1}\right]\, , \end{equation*}
where both \( J \) and \(\gamma\) need to be positive. Notice that, even though the $W$ state requires a bond dimension $N$ to be represented with diagonal matrices, we can find a Hamiltonian with interaction length $2$ which has this state as a ground state. 

{\it iii)} {\it $X_N$ states.---} Let us consider the family for \( N\geq 4 \) given by
\begin{equation*} 
|X_N(z)\rangle= (N-1)\left[|1\rangle^{\otimes N}+ z^{N-1} \sqrt{N} |W_N\rangle\right]\, ,
\end{equation*} 
up to normalization. The DMPS representation of this state is given by the $(N-1)$-dimensional matrices  $A_0=z\, \mathrm{diag}\left\{1,e^{2\pi i/(N-1)},e^{4\pi i/(N-1)}\ldots,e^{2\pi i(N-2)/(N-1)}\right\}$ and $A_1= \mathrm{diag}\left\{1,1,1,\ldots1\right\}$. In this case, the relevant interaction length is seen to be $3$. In fact, the parent Hamiltonian family reads
\begin{eqnarray*}
&H_X  = \frac{J_z}{3}\sum_{i=1}^N \left ( 3 +3 \sigma_i^z \sigma_{i+1}^z \sigma_{i+2}^z - \sigma_i^z \sigma_{i+1}^z- \sigma_i^z \sigma_{i+2}^z \right.\\
& \left.\ - \sigma_{i+1}^z \sigma_{i+2}^z - \sum_{j=0}^2\sigma_{i+j}^z\right) + \frac{J}{4}\sum_{i=1}^N \left(1-\boldsymbol{ \sigma}_i\cdot\boldsymbol{ \sigma}_{i+1}\right)\,,
\end{eqnarray*}
for any $J>0$ and $J_z>0$.

In conclusion, the proposed MPS classification fulfills all the aforementioned criteria for a versatile grouping of SLOCC classes, while maintaining a linear growth in the number of families with respect to the number of parties. At the same time, the unveiled nesting property allows us to use, in a scalable manner, the effort invested in the MPS classification for $N$ parties in the subsequent $N+1$ case. Additionally, we have provided a physical meaning for the proposed MPS classes by connecting them to paradigmatic properties in condensed matter. This missing link in the theory of entanglement classification has been exemplified for the case of the symmetric subspace, where the DMPS representation was used. It is noteworthy to mention that the provided DMPS classification can be formally extended from the symmetric subspace to more general sets of quantum states. We believe that MPS-based entanglement classifications will be able to connect mathematical aspects already known in quantum information theory with relevant physical features in many-body systems of experimental interest, such as quantum state preparation or the emergence of fractional magnetizations.

\section*{AUTHOR CONTRIBUTIONS}
M.S., as the first author, has been responsible for the development of this work. M.S. and I.L.E., with the support of R.D.C. and H.S., have made the mathematical demonstrations, carried out calculations, and provided examples. L. L. and E. S. suggested the seminal ideas that were improved by M.S. and I.L.E. All authors have carefully proofread the manuscript. E.S. proposed the project and supervised it throughout all stages.

\section*{ACKNOWLEDGMENTS}
We thank G\'eza T\'oth, G\'eza Giedke, Rom\'an Or\'us, Norbert Schuch and Jens Siewert for useful discussions.
The authors acknowledge support from Spanish MINECO FIS2012-36673-C03-02; Ram\'on y Cajal Grant RYC-2012-11391; UPV/EHU UFI 11/55 and EHUA14/04; Basque Government IT472-10 and IT559-10; PROMISCE, and SCALEQIT EU projects; European Social Fund; and the state budget of the Czech Republic, Project No. CZ.1.07/2.3.00/30.0041.
 
\section*{ADDITIONAL INFORMATION} 
The authors declare no competing financial interests.

\appendix

\section{The symmetric subspace and SLOCC classes}
We consider the symmetric subspace of \( N \)-particle systems, i.e. \( \mathrm{Sym}\left( \mathcal{H}^{\otimes N}\right) \).  For the sake of clarity, we say that a state is symmetric if and only if \( F |\psi\rangle= |\psi\rangle \), where \( F \) is a  flip or exchange operator, or the representation of any element of the permutation group. If the one-particle Hilbert space $\mathcal{H}$ has dimension \( d \) (i.e. qudits), we have
\[
\mathrm{dim}\left[\mathrm{Sym}\left( \mathcal{H}^{\otimes N}\right) \right]= {{N+d-1}\choose N}\,.
\]
In particular, the dimension of the symmetric subspace in an $N$-qubit system is \( N+1 \). Our goal is to organise in some physically relevant way the orbits of the group of invertible local operator (ILO) transformations, i.e. the orbits under transformations of the form
\begin{equation}
|\psi_N\rangle\to A^{\otimes N}|\psi_N\rangle\,, 
\end{equation}
where \( A \) is an invertible local operator acting on \( \mathcal{H} \). The reason for considering orbits under the group of ILOs is that stochastic local operators and classical communication operators (SLOCC) are constructed as ILOs~\cite{DVC00}. Notice that, in the symmetric case, we only need to consider ILOs of the form above, i.e. such that the invertible local operator acting on one party is the same for all parties~\cite{PhysRevA.81.052315}. 
 
\section{Diagonal MPS representation}
The main result of the paper relies on the diagonal matrix product state (DMPS) representation of a symmetric state. In this section, we give two independent proofs that any symmetric state can be written in this form.

\subsection{First proof: Majorana representation}
Our starting point is the \emph{Majorana} representation: for  \( N \)-qubit systems any symmetric state can be written as
\begin{equation}\label{eq:Majorana}
|\psi_N\rangle\propto \sum_{1\leq {i_1}\neq {i_2}\neq\cdots\neq {i_N}\leq N}|e_{i_{1}}\rangle\otimes|e_{i_{2}}\rangle\otimes\cdots\otimes|e_{i_{N}}\rangle\,,
\end{equation}
where 
\begin{equation}
|e_i\rangle= \alpha_i|0\rangle+ \beta_i|1 \rangle
\end{equation} are single qubit states. In the case of qubits, this decomposition is unique up to rearrangement of indices. Then, Eq.  (\ref{eq:Majorana}) can be rewritten as 
\begin{equation}
|\psi_N\rangle\propto \sum_{\sigma \in P_N} \sum_{k=0}^{N-1} \frac{(-1)^k}{(N-k)! k!} \left ( \sum_{l=1}^{N-k} |e_{\sigma(l)} \rangle \right )^{\otimes N}\,.
\end{equation}
This state is of the form
\begin{equation}\label{eq:diagonalstate}
|\psi_N\rangle\propto \sum_{k=1}^D |x_k\rangle^{\otimes N}\,,
\end{equation}
for some $D$, that, as we will see, can be easily cast in a diagonal MPS representation.

\subsection{Second proof: Vandermonde determinant}\label{ssec:vandermonde}
We give now an alternative proof of the existence of diagonal presentations (\ref{eq:diagonalstate}) for qubits. The dimension of the symmetric subspace is \( N+1 \), thus, if we can show that \( N+1 \) linearly independent vectors of the form \( |x_k\rangle^{\otimes N} \) exist, we have shown that any state can be written as their linear combination. 
We will find useful in what follows to relate states to their expansion in terms of Dicke states:
\begin{equation}
|\psi_N\rangle=\sum_{k=0}^{N} d_k |D_k^{(N)}\rangle\,,
\end{equation}
where these qubit Dicke states are defined by the number of ``excitations" \( k \), i.e., in the standard basis \( \{|0\rangle,|1\rangle\} \), the number of times \( |1\rangle \) appears, and they are normalised.

First, we prove the following statement:
\begin{lem}\label{le:diagMPS}
The set of \( N \)-particle states \( \{|x_k\rangle^{\otimes N}\}_{k=1}^{N+1} \) is linearly independent and generates the whole symmetric space if and only if the one-particle states \( \{|x_k\rangle\}_{k=1}^{N+1} \) are \emph{pairwise linearly independent} one particle states. 
\end{lem}
\begin{proof}
Let us write the one-site states in the standard basis:
\begin{equation}
|x_k\rangle= \beta_k |1\rangle + \alpha_k|0\rangle
\end{equation}
with \( | \alpha_k|^2+| \beta_k|^2=1 \). At most one of the \( \beta_k \)'s can be zero, since otherwise we would have the state \( |1\rangle \) repeated, and the set \( \{|x_k\rangle\}_{k=1}^{N+1} \) would not be pairwise linearly independent. On the other side of the equivalence, if more than one of the \( \beta_k \)'s are zero, the set of one-particle states  \( \{|x_k\rangle\}_{k=1}^{N+1} \) are not pairwise linearly independent. Let us assume first that none of the \( \beta_k \)'s are zero. 

The \( N \)-site state can be written as
\begin{eqnarray}
|x_k\rangle^{\otimes N}&=& \left(\beta_k |1\rangle + \alpha_k|0\rangle\right)^{\otimes N} \nonumber \\ 
&=& \sum_{n=0}^N \sqrt{{N\choose n}}\, \beta_k^n\alpha_k^{N-n} |D_n^{(N)}\rangle\,.
\end{eqnarray}
Linear independence of the set \( \{|x_k\rangle^{\otimes N}\}_{k=1}^{N+1} \) means that the only set of \( N+1 \) coefficients $\lambda_k$ such that \( \sum_{k=1}^{N+1} \lambda_k|x_k\rangle^{\otimes N}=0 \) is $\lambda_k=0$ $\forall k$. However, any linear superpositions of   \( \{|x_k\rangle^{\otimes N}\}_{k=1}^{N+1} \)  can be written in terms of the Dicke states as
\begin{equation}
\sum_{k=1}^{N+1} \lambda_k|x_k\rangle^{\otimes N}= \sum_{n=0}^N \gamma_n |D_n^{(N)}\rangle\,,
\end{equation}
with the following relation between the \(\{\gamma_i\}\) and the \(\{\lambda_i\}\) coefficients:
\begin{equation}
\gamma_n= \sqrt{{N\choose n}}\, \left(\sum_{k=1}^{N+1}\beta_k^n\alpha_k^{N-n} \lambda_k\right)\,.
\end{equation}
This means
\begin{eqnarray}
\begin{pmatrix}
\gamma_0\\ \gamma_1\\ \vdots\\ \gamma_N
\end{pmatrix}
&=& S V B 
\begin{pmatrix}
\lambda_1\\ \lambda_2\\ \vdots \\ \lambda_{N+1}
\end{pmatrix},
\end{eqnarray}
where $S=\mathrm{diag}\left(1, \sqrt{{N \choose 1}}, \sqrt{ {N \choose 2}},\ldots,1\right)$, $B=\mathrm{diag}\left(\beta_1^N, \beta_2^N, \ldots,\beta_{N+1}^N\right)$, and 
\begin{eqnarray}
V=
\begin{pmatrix}
\left(\frac{\alpha_1}{\beta_1}\right)^N& \left(\frac{\alpha_2}{\beta_2}\right)^N &\cdots & \left(\frac{\alpha_{N+1}}{\beta_{N+1}}\right)^N\\
\left(\frac{\alpha_1}{\beta_1}\right)^{N-1}&\left(\frac{\alpha_2}{\beta_2}\right)^{N-1}&\cdots & \left(\frac{\alpha_{N+1}}{\beta_{N+1}}\right)^{N-1}\\ 
\vdots & \vdots & \ddots&\vdots\\
1& 1 &\cdots & 1
\end{pmatrix}.
\end{eqnarray}
The condition for linear independence of the set \( \{|x_k\rangle^{\otimes N}\}_{k=1}^{N+1} \) is thus equivalent to the condition for $V$ to be invertible. Notice that $V$ is a Vandermonde matrix, and its determinant is
\begin{equation}
|V|= \prod_{1\leq k< l \leq N+1} \left[ \frac{ \alpha_l}{ \beta_l}-\frac{ \alpha_k}{ \beta_k} \right]\,,
\end{equation}
which is different from zero if and only if the set \( \{|x_k\rangle\}_{k=1}^{N+1} \) is pairwise linearly independent. We have thus proven the lemma in the case that none of the \( \beta_k \) coefficients are zero.

Let us now consider the possibility that one of the \( \beta \) is zero, e.g. $\beta_1=0$.  Then, we have w.l.o.g. \(\alpha_1=1\), and
\begin{eqnarray}
\begin{pmatrix}
\gamma_0\\ \gamma_1\\ \vdots\\ \gamma_N
\end{pmatrix}&=&  S\tilde V \tilde B
\begin{pmatrix}
\lambda_1\\ \lambda_2\\ \vdots \\ \lambda_{N+1},
\end{pmatrix},
\end{eqnarray}
where $\tilde B=\mathrm{diag}\left(1, \beta_2^N, \ldots,\beta_{N+1}^N\right)$ and 
\begin{eqnarray}
\tilde V=
\begin{pmatrix}
1& \left(\frac{\alpha_2}{\beta_2}\right)^N &\cdots & \left(\frac{\alpha_{N+1}}{\beta_{N+1}}\right)^N\\
0&\left(\frac{\alpha_2}{\beta_2}\right)^{N-1}&\cdots & \left(\frac{\alpha_{N+1}}{\beta_{N+1}}\right)^{N-1}\\ 
\vdots & \vdots & \ddots&\vdots\\
0& 1 &\cdots & 1
\end{pmatrix}.
\end{eqnarray}
Linear independence of the set at hand is equivalent to the invertibility of $\tilde V$, which reads
\begin{eqnarray}
|\tilde V| & = \begin{vmatrix}
\left(\frac{\alpha_2}{\beta_2}\right)^{N-1}&\cdots & \left(\frac{\alpha_{N+1}}{\beta_{N+1}}\right)^{N-1}\\ 
\vdots & \ddots&\vdots\\
1 &\cdots & 1
\end{vmatrix} \nonumber \\
&=\prod_{2\leq k< l \leq N+1} \left[ \frac{ \alpha_l}{ \beta_l}-\frac{ \alpha_k}{ \beta_k} \right]\,.
\end{eqnarray}
Since we have required that all \(\beta\) be different from zero, save for \(\beta_1\), we recover again the requirement that the one site vectors  \( \{|x_k\rangle\}_{k=1}^{N+1} \) be pairwise linearly independent, thus completing the proof.
\end{proof}

\begin{cor}\label{cordiag}
If \(   M+1 \geq D \), then  the set of $M$-particle states $\{|x_k\rangle^{\otimes M}\}_{k=1}^D$ is linearly independent if and only if the one-particle states $\{| x_k\rangle\}_{k=1}^D$ are pairwise linearly independent.
\end{cor}

It follows that any symmetric state can be written in the form of Eq.~\eqref{eq:diagonalstate} with $D\leq N+1$, where  the vectors $\{|x_{k\rangle}\}_{k=1}^{N+1}$ are pairwise linearly independent.
This gives us also a bound for the \emph{optimal} $D$, i.e. for the minimal \( D \) in the decomposition of a state. In what follows, we use  \( D(\psi) \) to denote the optimal \( D \) of the state \( |\psi\rangle \), if more than one state is being discussed. Otherwise, we normally denote the optimal value by  \( D \) without further qualification. Thus, as a corollary of Lemma \ref{le:diagMPS} we conclude that for all \( N \) qubit symmetric states \( D\leq N+1 \).

\subsection{Diagonal Matrix Product States (DMPS)}
In the two previous subsections, we have shown that a representation of the form given by Eq.~(\ref{eq:diagonalstate}) is always possible for any symmetric state of \( N \) qubits. In this subsection, we show that this is equivalent to a diagonal MPS representation.

Any quantum state $|\Phi \rangle \in \mathbb{C}^{2^N}$ admits an MPS representation given by 
\begin{equation}\label{eq:MPS}
|\Phi \rangle = \sum_{\mu_1, \ldots , \mu_N  = 1}^2 \mathrm{Tr}\, [A_{\mu_1}^{[1]} \cdots A_{\mu_N}^{[N]}]  \, |\mu_1 \ldots \mu_N \rangle
\end{equation}
where the set of complex matrices $\mathcal{K} = \{ A_{i_k}^{[k]} \in \mathcal{M}_D (\mathbb{C}), i_k = 1, 2 , \,  1 \leq k \leq N \}$ are called \emph{Kraus operators} or \emph{MPS matrices} and $D$ is the \emph{bond dimension} of the MPS \cite{PhysRevLett.91.147902}.

If the state is translational invariant, then there always exists a site--independent MPS representation such that  $A_{i_k}^{[k]} = A_i$ $\forall k$, normally by increasing the bond dimension of the MPS \cite{PGVWC07}[Theorem 3]. Moreover, two sets of Kraus operators $\{A_i \}$ and $\{ B_i \}$ are equivalent \cite{PGVWC07}[Theorem 2], i.e. they represent the same state, if there exists an invertible matrix $X$ such that $B_i = X A_i X^{-1}$, $\forall i$. 

It is clear that MPS representations implemented by diagonal matrices lead to symmetric states. That the opposite also holds is suggested in \cite{PGVWC07}[Apendix A.1]. We state it as a Theorem, whose proof is immediate given the previous constructions.

\begin{thm} \label{thm:optimalbond}
Any $N$-qubit permutation invariant state $|\Psi_{S}^{(N)}\rangle = \sum_{k=1}^D c_k |x_k \rangle^{\otimes N}$ has a diagonal MPS representation, with bond dimension $D \leq N+1$, given by the Kraus operators 
\begin{equation}\label{eq:diagKraus}
A_{\mu}  =  \sum_{k=1}^{D}  (c_k)^{\frac{1}{N}} \langle \mu | x_k \rangle \, | k \rangle \langle k |\,,
\end{equation}
where $D$ is the optimal bond dimension if and only if the decomposition $|\Psi_{S}^{(N)}\rangle = \sum_{k=1}^D c_k |x_k \rangle^{\otimes N}$ uses also the minimal number of vectors.
\end{thm}

\begin{proof}
Let us replace the Kraus operators given by  (\ref{eq:diagKraus}) in (\ref{eq:MPS}):
\begin{eqnarray*}
|\Psi_{S}^{(N)}\rangle & = &\sum_{\mu_1,\ldots, \mu_N =1}^2 \mathrm{Tr}\, [A_{\mu_1} \cdots A_{\mu_N}] |\mu_1 \ldots \mu_N \rangle \\
& = & \lefteqn{\sum\limits_{\substack{\mu_1,\ldots, \mu_N \\  k_1, \ldots , k_N}} | \mu_1 \ldots \mu_N  \rangle \langle \mu_1 \ldots \mu_N | x_{k_1} \ldots x_{k_N}\rangle \times} \\
& & \times c_{k_1}^{\frac{1}{N}} \cdots c_{k_N}^{\frac{1}{N}} \delta_{k_1 k_2} \delta_{k_2 k_3} \ldots \delta_{k_N k_1} \\
& = & \sum_{k=1}^{D} c_k  |x_k\rangle^{\otimes N}.
\end{eqnarray*}
The relation between the optimal bond dimension and the minimal number of vectors in the representation is trivial by construction. This concludes the proof.
\end{proof}

After this existence check, it now behoves us to ask for the optimal bond dimension for a given state is, and after its  properties.

\section{Entanglement and optimal bond dimension}

In this section, we justify the definition of entanglement family  by proposing and discussing the basic properties that we require. Moreover, we compare this definition with those in the literature, and we prove that our classification is stable with respect to changes in the number of parties. 

\subsection{Definition of family and comparison with previous works}

W. D\"ur, G. Vidal and J. I. Cirac established in 2000 that the relevant notion for entanglement classes is equivalence under SLOCC~\cite{DVC00}. Moreover, they proved that SLOCC transformations are implemented by ILOs. Indeed, $|\psi\rangle$ and $|\phi\rangle$ are connected via an SLOCC if and only if a set of local invertible operators $\{A_i\}_{i=1}^N$ exists, such that
\begin{equation}\label{SLOCC}
|\phi\rangle= \left(\bigotimes_{i=1}^N A_i\right)|\psi\rangle\,.
\end{equation}
If we are restricted to the symmetric subspace,  the existence of symmetric ILOs connecting two states  $|\psi\rangle$ and $|\phi\rangle$   is enough to ascertain their SLOCC equivalence; i.e. they are SLOCC equivalent if there exists a local invertible operator $A$ such that $|\phi\rangle=A^{\otimes N}|\psi\rangle$~\cite{PhysRevA.81.052315}.  

One should notice that Local Unitary (LU) equivalence is included in SLOCC equivalence. Thus, any LU entanglement classification is fully included in one and just one SLOCC entanglement class. In this sense, and according to the definition of ``family" that follows, we can say that SLOCC entanglement classes are actually \emph{families} with respect to LU classes.
\begin{defn}[Entanglement family] An entanglement family in a set of states is a collection of entanglement classes for the same set with the following properties: \begin{itemize}
\item
 An entanglement class fully belongs to one and only one family (SLOCC structure preservation)
\item
The class of separable states must be in a family of its own.
 \item
 The number of families grows in a controllable manner with the number of particles.
 \item
 There is a mathematical criterion for the arrangement of classes in families, with a physical interpretation.
\end{itemize}
Optionally, one might add an additional desideratum:
\begin{itemize}
\item
Independence of the criterion with respect to the number of parties \( N \).
\end{itemize}
\end{defn}
As stated, the SLOCC classification can then be understood as an arrangement in entanglement families of the LU entanglement classes. Even so, in what follows we will refer to SLOCC classes, because they are more relevant in what regards to quantum information tasks. In the following, we propose our sectioning into families and compare it to other previous classifications. 

\subsubsection{ DMPS classification}
We propose that the optimal \( D \) diagonal bond dimension be the criterion for classification.  It is well known, from the form  (\ref{eq:diagonalstate}), that it is an SLOCC invariant, and it has been proposed as a measure of entanglement \cite{EB01}. Here, we examine its value as a classification tool.

As \( D \) is an SLOCC invariant, an entanglement SLOCC class belongs to one family, and just one.  All separable states have \( D=1 \), and if a state presents \( D=1 \) then it is separable, by definition. Thus the separable class lies in a family of its own. As we see, the first two criteria we have proposed for families are met. In the next sections, we will discuss the physical interpretation of our classification, and its dependence on $N$. 

\subsubsection{Comparison to the degeneracy classification}
The organisation in families proposed in T. Bastin {\it et al}~\cite{BKMGLS09} relies on the Majorana representation. We assert that this classification in families meets our requirements. For all \( N \) the family \( \mathcal{D}_N \) includes the separable class and only the separable class. The  degeneracy configuration which is the classification criterion is an SLOCC invariant. The third property is elucidated in L. Lamata {\it et al}~\cite{PhysRevA.87.032325}. Instead, the fourth tentative desideratum is not so clearly fulfilled. 

Notice that this classification is different from ours. In fact, already for \( N=4 \)  their \( \mathcal{D}_{1,1,1,1} \) family includes classes that would go to our \( D=2 \) and \( D=3 \) families.

 \subsubsection{Comparison to SLOCC normal form classification}
In F. Verstraete {\it et al} \cite{PhysRevA.68.012103}, they give a procedure for constructing normal forms under SLOCC operations. It is a classification of SLOCC classes into families, since all the elements of a SLOCC class will go into the same normal form. It does not, however, meet all our desiderata: the states in the \( W_N \) class have normal form equal to zero, which is also the normal form for separable states. Thus, the class of separable states is not in a family of its own.

\subsubsection{Comparison to entanglement polytopes classication}
In Ref. \cite{WDGC13}, Walter et al provide an organization of SLOCC classes in families according to polytopes which classify the one particle reduced density matrix. For the specific case of symmetric states of \( N \) qubits, which they compute, they obtain \( \lceil{N/2}\rceil+1 \) families; i.e. growing linearly with the number of qubits, thus fullfilling our requirement of controlled growth of the number of families. Separable states however belong to all entanglement polytopes (although this might be taken as quibbling, since they are the only states for which the relevant parameter is 1). An operational characterization is given by the maximal linear entropy of entanglement achievable in a family.
 
\subsection{Nesting and stability}
In this subsection, we prove a limited Theorem on nesting, in order to clarify how our classification is stable with respect  to changes of $N$. Indeed, we define the stability property as follows: the family is \emph{stable} if, when moving up in number of parties, one maintains the minimal bond dimension and thus the local properties. To clarify: since we consider the family for \( N \) qubits defined in terms of the optimal bond dimension of its MPS representations, we move up in number of parties by starting with an \( N \)-qubit state in the family, fixing its matrices \( A_0 \) and \( A_1 \), and now using the selfsame matrices for \( N+1 \) qubits. This can be done for all states in the family. We prove now that stability is a \emph{generic} property, in  that it holds  for  \( D\leq\lfloor{N/2}\rfloor+1 \), given \( N \). We shall explain later (Subsection \ref{sec:cattle}) in what sense this  is indeed generic. Since we are dealing with minimal bond dimensions for different states (in different Hilbert spaces, even), we explicitly denote the optimal bond dimension for state \( |\psi\rangle \) by \( D(\psi) \) in this context.
\begin{thm}[Nesting]
Let us consider an $N$-particle symmetric state of qubits $|\psi_N\rangle=\sum_{k=1}^{D(\psi_N)} |x_k\rangle^{\otimes N}$ with optimal bond dimension $D(\psi_N)$, such that $D(\psi_N)\leq \lfloor N/2\rfloor +1$. Then, the state $|\psi_{N+1}\rangle=\sum_{k=1}^{D(\psi_N)}|x_k\rangle^{\otimes (N+1)}$ has optimal bond dimension $D(\psi_{N+1})=D(\psi_N)$.
\end{thm}

Let us introduce the {\it Schmidt binary rank} at level \( M \) for the studied systems, which we will denote, for an $N$-particle symmetric state \( |\psi_N\rangle \), as \( \mathrm{sbr}_M\left(\psi_N\right) \). It is defined as the Schmidt number for the bipartition of the system in \( M \) and \( N-M \) qubits, i.e. as the minimal number \( s \) such that \( |\psi_N\rangle \) can be written in the form
\begin{equation}\label{sbr}
|\psi_N\rangle= \sum_{k=1}^s |\xi_k\rangle_M |\chi_k\rangle_{N-M}\,,
\end{equation}
with \( \{|\xi_k\rangle_M\}_{k=1}^s \)  orthonormal vectors in the space of \( M \) qubits, and similarly \(  \{|\chi_k\rangle_{N-M}\}_{k=1}^s \) in the space of $N-M$ qubits. The following Lemmas are useful to prove the Nesting Theorem.
\begin{lem}\label{sbr1}
$\mathrm{sbr}_M\left(\psi_N\right)\leq D(\psi_N)\,$ for all $1\leq M<N$. 
\end{lem}
\begin{proof}
We can write $|\psi_N\rangle=\sum_{k=1}^{D(\psi_N)}|x_k\rangle^{\otimes N}$. By defining \( | \xi_k\rangle_M=|x_k\rangle^{\otimes M} \) and \( | \chi_k\rangle_{N-M}=|x_k\rangle^{\otimes (N-M)} \), and comparing with Eq.~\eqref{sbr}, it directly follows that, for all \(1\leq M<N \) and all symmetric \( N \)-particle states, $\mathrm{sbr}_M\left(\psi_N\right)\leq D(\psi_N)\,$.
\end{proof}

\begin{lem}\label{sbrlin}
If $|\psi_N\rangle= \sum_{k=1}^D |\xi_k\rangle_{M} |\chi_k\rangle_{N-M}$, where $\{|\xi_k\rangle_M\}_{k=1}^D$ and $\{ |\chi_k\rangle_{N-M}\}_{k=1}^D$ are two sets of linearly independent vectors, then $\mathrm{sbr}_M(\psi_N)=D$.
\end{lem}
\begin{proof}
There exist two invertible operators $X_{A,B}$, such that the sets $\{|\tilde \xi_k\rangle_M\}_{k=1}^D$ and $\{ |\tilde \chi_k\rangle_{N-M}\}_{k=1}^D$, where $|\tilde \xi_k\rangle_{M}= X_A|\xi_k\rangle_M$ and $|\tilde \chi_k\rangle_{M}= X_{B}|\chi_k\rangle_{N-M}$, are orthogonal. This implies that  $\mathrm{sbr}_M(\tilde \psi_N)=D$ for the state $|\tilde \psi_N\rangle=X_A\otimes X_B |\psi_N\rangle$. As the Schmidt binary rank is invariant under local invertible operations, we have that $\mathrm{sbr}_M(\psi_N)=\mathrm{sbr}_M(\tilde \psi_N)=D$.
\end{proof}

\begin{proof}[Proof of Theorem 2]
We have that $D(\psi_{N+1})\leq D(\psi_N)$ by construction. On the other hand, if \( D(\psi_N)\leq \lfloor N/2\rfloor+1 \), then   an integer  $M$ exists such that $D(\psi_{N})\leq M+1$ and $D(\psi_N)\leq N+1-M$. As $\{|x_k\rangle\}_{k=1}^{D(\psi_N)}$ are necessarily pairwise linearly independent, Corollary~\ref{cordiag} implies that, in the bipartition \( |\psi_{N+1}\rangle=\sum_{k=1}^{D(\psi_N)}\left(|x_k\rangle^{\otimes M}\right)\otimes \left(|x_k\rangle^{\otimes (N+1-M)}\right) \), both sets \( \left\{|x_k\rangle^{\otimes M}\right\}_{k=1}^{D(\psi_N)} \) and  \( \left\{|x_k\rangle^{\otimes (N+1-M)}\right\}_{k=1}^{D(\psi_N)} \) are linearly independent.Then, Lemma~\ref{sbrlin} ensures that \( \mathrm{sbr}_M\left(\psi_{N+1}\right)=D(\psi_N) \).   Lemma~\ref{sbr1}  now implies that $D(\psi_{N})\leq D(\psi_{N+1})$, which concludes the proof.
\end{proof}

\subsection{Examples of DMPS representation}
In this subsection, we provide  examples of  optimal matrix product representations. This comes with a general method to find the optimal bond dimension for a given state.

\subsubsection{Statement of the problem}
Let us consider an \( N \)-particle symmetric state expressed in the Dicke basis
\begin{equation}
|\psi\rangle=\sum_{\alpha=0}^{N} d_\alpha |D_\alpha^{(N)}\rangle\,,
\end{equation}
as our initial data. In order to identify a Matrix Product representation, we shall use the equivalence of Theorem \ref{thm:optimalbond}. Thus, we look for \( D \) one-qubit states of the form $|x_k\rangle=x_k|0\rangle + y_k |1\rangle$,
such that we can write $|\psi\rangle= \sum_{k=1}^D |x_k\rangle^{\otimes N}$. This demands
\begin{equation}\label{eq:conditionsD}
\sum_{k=1}^D x_k^{N-\alpha} y_k^\alpha ={N\choose \alpha}^{-1/2}\, d_\alpha,
\end{equation} 
for all \(\alpha\) from \( 0 \) to \( N \).  The task at hand is to examine, given the state, which is the minimum \( D \) for which the set of equations~\eqref{eq:conditionsD} has a solution, and, once having identified such minimum \( D \), to obtain such solution.

It is important to notice that, even if we were to start by positing two \( D\times D \) diagonal matrices \( A_0 \) and \( A_1 \), which do determine a symmetric \( N \) qubit state, we would not have determined whether that bond dimension is indeed optimal for the state. By construction, the optimal value would be smaller than or equal to the bond dimension proposed, and it is still pertinent to examine the set of solutions of Eq.~\eqref{eq:conditionsD}. 

In the following examples, we give optimal \( D \) constructions, that are clearly not unique for a given state.

\subsubsection{ \( W_N \) state}
The $W_N$ state is defined as
\begin{eqnarray*}
|W_N\rangle = & \frac{1}{ \sqrt{N}} (|100\ldots000\rangle+|010\ldots000\rangle+\cdots +  \\ 
+ & |000\ldots010\rangle+|000\ldots001\rangle)\,.
\end{eqnarray*}
with \( N\geq3 \). Equivalently, this is the Dicke state with one excitation, namely \( |D_1^{(N)}\rangle \). The coefficients \( d_\alpha \) to be inserted in Eq.~(\ref{eq:conditionsD}) are thus \( d_\alpha= \delta_{ \alpha1} \), with \(\delta_{ \alpha \beta}\) the 
Kronecker delta. 

Let us write a solution in terms of matrices:
\begin{widetext}
\begin{eqnarray}
A_0&= \lambda\, \mathrm{diag}\left\{1,e^{2\pi i/N(N-1)},\ldots,e^{2\pi i k/N(N-1)},\ldots e^{2\pi i(N-2)/N(N-1)},0\right\}\\
A_1&=\mu\, \mathrm{diag}\left\{1,e^{2\pi i(N-1)/N},\ldots,e^{2\pi i(N-k)/N},\ldots,e^{4\pi i/N},\left(1-N\right)^{1/N}\right\}.
\end{eqnarray}
\end{widetext}
This is adequate since, for \( 0\leq l<N \), we have
\begin{eqnarray}
\frac{\mathrm{Tr}\left[A_0^{N-l}A_1^l\right]}{ \lambda^{N-l} \mu^l} &=&\sum_{k=0}^{N-2} e^{2\pi i k(N-l)/N(N-1)} e^{2\pi i l(N-k)/N}\nonumber \\
&=&e^{2\pi i l}\sum_{k=0}^{N-2}\exp\left[2\pi i k(1-l)/(N-1)\right]\nonumber \\
&=& \frac{1-\exp\left[2\pi i (1-l)\right]}{1-\exp\left[2\pi i (1-l)/(N-1)\right]} \nonumber \\
&=&\delta_{l1} (N-1)\,,
\end{eqnarray}
and \( \mathrm{Tr}\left[A_1^N\right]=0 \). Therefore, we can choose any $\lambda$, $\mu$ respecting the normalisation relation \( \lambda^N ( \mu/ \lambda) (N-1)= 1/ \sqrt{N} \). Additionally, we can also fix the gauge, by means of \( \sum_i A_i^{\dag} A_i\propto 1\). One can now write this solution in terms of \( |x_k\rangle \) states using Eq.~\eqref{eq:diagKraus} with $c_k=1$:
 \begin{eqnarray*}
  |x_k\rangle&=&\mu |1\rangle+ e^{2\pi i (k-1)/(N-1)}\lambda |0\rangle\,, \; 1\leq k\leq N-1\,\\
 |x_N\rangle&=&\mu(1-N)^{1/N}|1\rangle\,
 \end{eqnarray*}
The existence of this form guarantees that  \( D(W_N)\leq N \). We need to prove that indeed \( D(W_N)=N \). In order to do that, consider the set of equations (\ref{eq:conditionsD}), particularised for this case, for a generic \( D \). Assume first that none of the coefficients \( x_\alpha \) is zero. Then, up to normalisation, the equations can be written in the equivalent form
\begin{equation}\label{optD}
\begin{pmatrix}
1&1&\cdots&1\\
z_1& z_2&\cdots&z_D\\
z_1^2&z_2^2&\cdots&z_D^2\\
\vdots&\vdots&\ddots&\vdots\\
z_1^N&z_2^N&\cdots&z_D^N
\end{pmatrix}
\begin{pmatrix}
x_1^N\\ x_2^N\\ \vdots\\ x_D^N
\end{pmatrix}= 
\begin{pmatrix}
0\\1\\0\\ \vdots\\0
\end{pmatrix}\,,
\end{equation}
where \( z_\alpha= y_\alpha/x_\alpha \). Let us define the $(N+1)\times D$ matrix 
\begin{equation}
Z_D=\begin{pmatrix}
1&1&\cdots&1\\
z_1& z_2&\cdots&z_D\\
z_1^2&z_2^2&\cdots&z_D^2\\
\vdots&\vdots&\ddots&\vdots\\
z_1^N&z_2^N&\cdots&z_D^N
\end{pmatrix}\,,
\end{equation}
where $z_\alpha$ are all distinct, and the \( N+1 \) component vector
\begin{equation}
b= \begin{pmatrix}
0\\1\\0\\ \vdots\\0
\end{pmatrix}\,.
\end{equation}
The requirement that the rank of \( Z_D \) be equal to the rank of the augmented matrix \( \left[Z_D\ b\right] \) is a necessary condition for the existence of solutions to the set of linear equations~\eqref{optD} for \( x_\alpha^N \), with $x_\alpha\not=0$ $\forall \alpha$. As \( \mathrm{rank}\left(Z_D\right)=D \) (because all the \( z_\alpha \) must be distinct), it is a necessary condition for existence of solutions of the system (\ref{eq:conditionsD}) that the rank of the augmented matrix  \( \left[Z_D\ b\right] \) also be \( D \).
 
 The case in which one of the coefficients \( x_\alpha=0 \) can be written in a similar manner. Without loss of generality, assume that the null coefficient is the \( D \)-th one. Then, the system of equations can be written as 
 \begin{equation}\label{optD2}
\begin{pmatrix}
1&1&\cdots&1\\
z_1& z_2&\cdots&z_{D-1}\\
z_1^2&z_2^2&\cdots&z_{D-1}^2\\
\vdots&\vdots&\ddots&\vdots\\
z_1^N&z_2^N&\cdots&z_{D-1}^N
\end{pmatrix}
\begin{pmatrix}
x_1^N\\ x_2^N\\ \vdots\\ x_{D-1}^N
\end{pmatrix}= 
b- \begin{pmatrix}
0\\0\\ \vdots\\ 0\\ y_D^N
\end{pmatrix}\,\equiv \tilde b_D.
\end{equation}
Then, a similar condition is retrieved, namely that the rank of \( Z_{D-1} \) and the rank of the augmented matrix are equal, which is a necessary condition for the existence of solutions for the system of equations.
  
Let us start by examining the case of \( D=N-1 \), initially for \( x_\alpha\neq0\ \forall \alpha \). The necessary condition for the existence of solutions is that all the \( N\times N \) minors of
 \begin{equation}
 \begin{pmatrix}
 1&\cdots&1&0\\
 z_1&\cdots&z_{N-1}&1\\
 z_1^2&\cdots& z_{N-1}^2&0\\
 \vdots&\ddots&\vdots&\vdots\\
 z_1^N&\cdots& z_{N-1}^N&0
 \end{pmatrix}
 \end{equation}
 are zero, using the hypothesis that for all \( \alpha\neq \beta \) it is the case that \( z_\alpha\neq z_\beta \). Thus, compute first 
 \begin{equation*}
 \begin{vmatrix}
 z_1&\cdots&z_{N-1}&1\\
 z_1^2&\cdots& z_{N-1}^2&0\\
 \vdots&\ddots&\vdots&\vdots\\
 z_1^N&\cdots& z_{N-1}^N&0
 \end{vmatrix}=(-1)^{N+1}\left(\prod_{\alpha=1}^{N-1} z_\alpha^2\right) \begin{vmatrix}1&\cdots&1\\ z_1&\cdots& z_{N-1}\\ \vdots&\ddots&\vdots\\ z_1^{N-2}&\cdots & z_{N-1}^{N-2}\end{vmatrix}
 \end{equation*}
 Since the last factor is a Vandermonde determinant, in order for this minor to be zero, it is necessary that one of the \( z_\alpha \) is zero. Without loss of generality, let that be \( z_{N-1}=0 \). We then compute
 \begin{eqnarray*}
 \begin{vmatrix}1&\cdots&1&1&0\\
 z_1&\cdots&z_{N-2}&0&1\\
 z_1^2&\cdots& z_{N-2}^2&0&0\\
 \vdots&\ddots&\vdots&\vdots&\vdots\\
 z_1^{N-1}& \cdots& z_{N-2}^{N-1}&0&0\end{vmatrix}
 &=&
  \begin{vmatrix}
 z_1^2&\cdots& z_{N-2}^2 \\
 \vdots&\ddots&\vdots \\
 z_1^{N-1}& \cdots& z_{N-2}^{N-1}\end{vmatrix}  \\ 
 &=&\left(\prod_{\alpha=1}^{N-2}z^2_\alpha\right)  \begin{vmatrix}1&\cdots&1\\ z_1&\cdots& z_{N-2}\\ \vdots&\ddots&\vdots\\ z_1^{N-3}&\cdots & z_{N-2}^{N-3}\end{vmatrix}\,.
 \end{eqnarray*} 
 In order for this minor to be zero as well, we would need another \( z_\alpha \) to be zero. That is however excluded,  since we require from the beginning that all the \( z_\alpha \) to be different. It follows that the two minors computed so far cannot be simultaneously zero, while keeping pairwise linear independence of the \( \{|x_\alpha\rangle\}_{\alpha=1}^{N-1} \). Since all \( N \times N\) minors must be zero for the rank of the augmented matrix to be \( N-1 \), and two of those cannot be simultaneously zero, under the assumption that all \( x_\alpha\neq0 \), it follows that the rank of the augmented matrix is larger than the rank of \( Z_{N-1} \) and there is no acceptable solution of this form.
 
 Let us now relax the assumption that all \( x_\alpha\neq0 \), and allow \( x_{N-1} =0\). Then, a similar analysis with the augmented matrix \( [Z_{N-2}\, \tilde b_{N-1}]\), where $\tilde b$ is defined in Eq.~\eqref{optD2}, shows again the inexistence of solutions.

The arguments above hold for any \( D<N-1 \), as the image of $Z_{D'}$ is included in the image of $Z_D$, if $D'< D$.  We have thus finally proven that the optimal bond dimension for \( W_N \) is indeed \( N \).

 By means of this example, we have provided an algorithm for the determination of the optimal bond dimension for a symmetric state of \( N \) qubits: 1) construct the \( b \) vector from the Dicke basis expansion of the state, and define \( \tilde b_{\tilde D} \equiv b -(0,0,\cdots,y_{\tilde{D}}^N)^T\); 2) start with \( \tilde{D}=1 \); 3) Is it possible for \( \mathrm{rank}[Z_{\tilde{D}}\ b] \) to be \( \tilde{D} \) or for \( \mathrm{rank}[Z_{\tilde{D}-1}\ \tilde b_{\tilde D}] \) to be \( \tilde{D}-1 \)? If yes, then \( D=\tilde{D} \). If no, 4) let \( \tilde{D} \) be set to  \( \tilde{D}+1 \) and return to 3.

\subsubsection{$\text{GHZ}_N$ state}
It is obvious by construction that the $\text{GHZ}_N$ states \( | \mathrm{GHZ}_N\rangle=\left(|00\cdots0\rangle+|11\cdots1\rangle\right)/ \sqrt{2} \) have \( D=2 \) and do not really require the machinery presented above. Nonetheless, it can be instructive to apply it to this case. We have to distinguish between two cases: a) \( N=2 \) and b) \( N>2 \). Let us first tackle \( N=2 \). In this case we have to study whether it is possible for 
\begin{equation}
\begin{pmatrix}
1&1&1\\ z_1& z_2&0\\ z_1^2& z_2^2&1
\end{pmatrix}
\end{equation}
to be of rank 2. This would require that the determinant is zero, i.e. \( \left(z_2-z_1\right)\left(1+z_1 z_2\right)=0 \), which, under the condition \( z_1\neq z_2 \) entails \( z_2=-1/z_1 \). Inserting this into the system of equations, we have
\begin{eqnarray}
x_1^2+ x_2^2&=& 1\,,\\
z_1 x_1^2 - \frac{x_2^2}{z_1}&=&0\,,\\
z_1^2 x_1^2+ \frac{x_2^2}{z_2^2}&=&1\,.
\end{eqnarray}
The second equation provides us with \( x_2^2= z_1^2 x_1^2 \). Substituting in the others we are led to \( x_1^2=1/2 \) and \( z_1^2=1 \), which in turn give us \( x_2^2=1/2 \). This is a valid solution, as $x_{1,2}\not=0$. We recover thus the following \( D=2 \) decomposition of the $\text{GHZ}_2$ state:
\begin{equation}
| \mathrm{GHZ}_2\rangle= \frac{1}{2 \sqrt{2}}\left(|0\rangle+|1\rangle\right)^{\otimes 2}+ \frac{1}{2\sqrt{2}}\left(|0\rangle-|1\rangle\right)^{\otimes 2}\,.
\end{equation}
Passing now to the general case, \( N\geq3 \), we first attempt \( D=2 \) with the hypothesis \( x_1,x_2\neq0 \). This entails examining the rank of 
\begin{equation}
\begin{pmatrix}
1&1&1\\ z_1& z_2&0\\ z_1^2& z_2^2&0\\
\vdots&\vdots&\vdots\\ z_1^N& z_N^3&1
\end{pmatrix}\,.
\end{equation}
We should study all \( 3\times 3 \) minors, in principle. It is however enough to look at just two: first
\begin{equation}
\begin{vmatrix}
1&1&1\\ z_1& z_2&0\\ z_1^2& z_2^2& 0
\end{vmatrix}=z_1 z_2(z_2-z_1)\,,
\end{equation}
which, without loss of generality, leads us to require \( z_2=0 \), since we require it to be zero, and that \( z_1\neq z_2 \). Under this \( z_2=0 \) condition we next examine
\begin{equation}
\begin{vmatrix}
1&1&1\\ z_1& 0& 0\\
z_1^N& 0&1
\end{vmatrix}=-z_1\,.
\end{equation}
For this to be zero as well we would have to require \( z_1=0 \), which we cannot allow by our requirement that \( z_1\neq z_2 \). Thus there is no \( D=2 \) presentation for \( | \mathrm{GHZ}_N\rangle \) with \( N\geq3 \) such that both  \( x_1 \) and \( x_2 \) are different from zero. Let us now assume \( x_2=0 \). We have to study the system
\begin{equation}
\begin{pmatrix}
1\\ z_1\\ \vdots\\ z_1^N
\end{pmatrix} x_1^N= \begin{pmatrix}
1\\0\\ \vdots\\ 1
\end{pmatrix}- \begin{pmatrix}
0\\ 0\\ \vdots\\ y_2^N
\end{pmatrix}\,.
\end{equation}
The solution is given by \( z_1=0 \), \( x_1^N=1 \), \( y_2^N=1 \), which corresponds to the trivial construction
\begin{equation}
| \mathrm{GHZ}_N\rangle= \frac{1}{ \sqrt{2}}\left(|0\rangle^{\otimes N}+ |1\rangle^{\otimes N}\right)\,.
\end{equation}

\subsubsection{ \( X_N \) state}
We now give an example of a family of symmetric \( N \)-particle  states ($N\geq 4$) that is explicitly in a different SLOCC class with respect the previous ones. Let us denote it by \( X_N(z) \). They are given by \( D=N-1 \) one-qubit vectors as follows:
\begin{equation}
|x_k\rangle=|1\rangle+ e^{2\pi i(k-1)/(N-1)}z |0\rangle\,,
\end{equation}
with \( k=1,\ldots, N-1 \) and $z\not=0$. The state is thus, by a non-normalised representative,
\begin{eqnarray}
|X_N(z)\rangle&=& \frac{1}{N-1}\sum_{k=1}^{N-1}|x_k\rangle^{\otimes N} \nonumber \\ &=& |1\rangle^{\otimes N}+ z^{N-1} \sqrt{N}|W_N\rangle\,.
\end{eqnarray}
The case \( N=4 \) (with \( z=2^{-1/6} \)) was introduced in A. Osterloh and J. Siewert~\cite{PhysRevA.72.012337} as a maximal entangled genuine four-partite state, with the notation \( |\Phi_2\rangle \). These authors also introduced the general case \( N \), with \( z= N^{-1/2(N-1)} \)~ \cite{1367-2630-12-7-075025}. 

In the following, we prove that the optimal bond dimension for $X_N(z)$ is $N-1$, by applying the same method  as in the $W_N$ and $\text{GHZ}_N$ cases. The case $D=N-1$ has a solution by construction. Let us examine the case of $D=N-2$, first for $x_\alpha\not=0$ $\forall \alpha$. We need that all the $(N-1)\ \times  (N-1)$ minors of
 \begin{equation}
 \begin{pmatrix}
 1&\cdots&1&0\\
 z_1&\cdots&z_{N-2}& z^{N-1}\\
 z_1^2&\cdots& z_{N-2}^2&0\\
 \vdots&\ddots&\vdots&\vdots\\
 z_1^N&\cdots& z_{N-2}^N&1
 \end{pmatrix}
 \end{equation}
are zero, in order to have solution. Let us compute first 
 \begin{eqnarray}
 \begin{vmatrix}
 z_1&\cdots&z_{N-2}&z^{N-1}\\
 z_1^2&\cdots& z_{N-2}^2&0\\
 \vdots&\ddots&\vdots&\vdots\\
 z_1^{N-1}&\cdots& z_{N-2}^{N-1}&0
 \end{vmatrix} &= & (-1)^{N}z^{N-1}\left(\prod_{\alpha=1}^{N-2} z_\alpha^2\right) \times \nonumber \\ 
& \times & \begin{vmatrix}1&\cdots&1\\ z_1&\cdots& z_{N-2}\\ \vdots&\ddots&\vdots\\ z_1^{N-3}&\cdots & z_{N-2}^{N-3}\end{vmatrix}\,.
 \end{eqnarray}
In order to be zero, it is necessary that one of the $z_\alpha$ be zero. Let us assume, without loss of generality, that $z_{N-2}$=0. We now compute
\begin{widetext}
\begin{equation}
 \begin{vmatrix}1&\cdots&1&1&0\\
 z_1&\cdots&z_{N-3}&0&z^{N-1}\\
 z_1^2&\cdots& z_{N-3}^2&0&0\\
 \vdots&\ddots&\vdots&\vdots&\vdots\\
 z_1^{N-2}& \cdots& z_{N-3}^{N-2}&0&1\end{vmatrix}
 =
 z^{N-1}\begin{vmatrix}
 z_1^2&\cdots& z_{N-3}^2 \\
 \vdots&\ddots&\vdots \\
 z_1^{N-2}& \cdots& z_{N-3}^{N-2}\end{vmatrix}=z^{N-1}\left(\prod_{\alpha=1}^{N-3}z^2_\alpha\right)  \begin{vmatrix}1&\cdots&1\\ z_1&\cdots& z_{N-3}\\ \vdots&\ddots&\vdots\\ z_1^{N-4}&\cdots & z_{N-3}^{N-4}\end{vmatrix}\,.
 \end{equation} 
 \end{widetext}
In order to be zero, we would need another $z_\alpha$ to be zero, but this case is excluded by hypothesis.
If we assume now that $x_{N-1}=0$, then a similar argument brings us to the conclusion that the optimal bond dimension is $N-1$.

\section{Parent Hamiltonians}
\subsection{Introduction}
Given a state \( \psi \) we say that a Hamiltonian \( H \) is a parent Hamiltonian of \( \psi \) if this state is a possible ground state of \( H \), not necessarily unique \cite{Mi11}. For finite-dimensional systems, a trivial parent Hamiltonian can be built by simply considering the projector onto the state \( \psi \), denoted by \( P_\psi \). In fact, for any positive \( \alpha \), the Hamiltonian \( H= \alpha\left(1-P_{ \psi}\right) \) has \(\psi\) as its only ground state. The concept of parent Hamiltonian comes in as useful when it can be written in terms of short range or very well-controlled long-range interactions. In the MPS representation, the considered state is written as locally as it is possible, and the minimal bond dimension for an MPS provides us with some control over the range of the Hamiltonian in question. One systematic way of searching for a given parent Hamiltonian is to consider the \( M \)-particle reduced density matrices, that, for some number \( M \) on, will not be full rank. It is enough to identify the projectors onto the kernels, and compute the intersection of their complements. The projector on the complement to that intersection is a parent Hamiltonian for the state at hand, with shorter interaction range.

\subsection{Parent Hamiltonians for symmetric states}
The case of symmetric states is particularly simple, since all \( M \)-particle reduced density matrices $\rho^{(M)}$ are identical. Let us give the name \( P_M \) to the \( M \)-particle projector onto that kernel of $\rho^{(M)}$. Notice that \( P_M \) is an \( M \) particle operator. Acting by translation on \( P_M\otimes \mathbf{1}^{\otimes(N-M)} \) (i.e. producing \( \mathbf{1}\otimes P_M\otimes \mathbf{1}^{\otimes(N-M-1)} \), \( \mathbf{1}^{\otimes2}\otimes P_M\otimes \mathbf{1}^{\otimes(N-M-2)} \), and succesively), and adding the translated terms, we obtain a parent Hamiltonian for the state at hand, with interaction length \( M \).

Let \( \rho^{(M)} \) denote the reduced density matrix for \( M \) qubits in this symmetric case, that is \( \rho^{(M)}= \mathrm{Tr}_{N-M}\left(|\psi_N\rangle\langle\psi_N|\right) \). Assume that \( |\psi_N\rangle \) has optimal bond dimension \( D \), i.e., it can be written as \( |\psi_N\rangle=\sum_{k=1}^D |x_k\rangle^{\otimes N} \) for some non-normalised vectors \( \left\{|x_k\rangle\right\}_{k=1}^D \), which are pairwise linearly independent.  We have that
\begin{equation}
\rho^{(M)}=\sum_{k,j=1}^D\left(\langle x_j|x_k\rangle\right)^{N-M} \left(|x_k\rangle\langle x_j|\right)^{\otimes M}\,.
\end{equation}
It follows that \( \mathrm{rank}\left[ \rho^{(M)}\right]\leq D \), since the image of $\rho^{(M)}$ is included in the span of \( \left\{|x_k\rangle^{\otimes M}\right\} \). This gives us a first bound on the interaction range. In fact, \( \rho^{(M)} \) acts on the space of \( M \) qubits, which is a linear space of dimension \( 2^M \). If the rank is smaller than the dimension, the kernel is not trivial. Clearly, if \( M>\log_2D \) it follows that the kernel of \( \rho^{(M)} \) is not trivial, and thus we know that the interaction length needed will be at most \( \lfloor\log_2 D\rfloor+1 \). 

Notice that the previous bound has been obtained allowing to project onto the whole linear space of \( M \) qubits, being the vectors symmetric or not. If we allow the projection onto the orthogonal part of the symmetric space, we can obtain a parent Hamiltonian whose interaction range is two, and whose ground state space includes the whole symmetric space. For instance, it is enough to consider the projection onto the state \( |01\rangle-|10\rangle \). However, our goal is to build parent Hamiltonians able to discriminate amongst symmetric states. Therefore, we derive a bound on the interaction length for non-trivial parent Hamiltonians, whose ground state space contains the orthogonal part of the symmetric space.

\subsection{Generic states and sharper bound for interaction lengths}\label{sec:cattle}

In the following, we prove the bound \( n\leq\lfloor{N/2}\rfloor+1 \) for the interaction length of parent Hamiltonians built considering only the symmetric space. We prove it by using a technique translated from algebraic geometry \cite{Landsberg:2010hl}. Given an \( N \)-qubit symmetric state \( |\psi_N\rangle \) we define for each \( M \) from 1 to \( N-1 \) the linear transformation \( \Psi_M \) that maps \( M \)-qubit  symmetric bra vectors, \(\langle\phi_M|  \), to \( (N-M) \)-qubit symmetric ket vectors by
 \begin{equation}
 \Psi_M\left(\langle\phi_M|\right)=\langle\phi_M|\psi_N\rangle\,.
 \end{equation}
Given bases in the space of symmetric states of \( M \) qubits, \( \{|e_M^k\rangle\}_{k=1}^{M+1} \), and of \( N-M \) qubits, \( \{|e_{N-M}^k\rangle\}_{k=1}^{N-M+1} \), the matrix elements for the linear transformation \( \Psi_M\) are computed to be
 \begin{equation}\label{eq:psimatrixel}
 \left(\Psi_M\right)_{lk}=\left(\langle e_{N-M}^l|\otimes\langle e_M^k|\right)|\psi_N\rangle\,.
 \end{equation}
 The reduced density matrix for \( N-M \) qubits is now written in the basis \( \{|e_{N-M}^k\rangle\}_{k=1}^{N-M+1} \) as
 \begin{equation}
 \rho^{(N-M)}_{lm}=\left(\Psi_M\right)_{lk}\left(\Psi^{\dag}_M\right)_{km}\,.
 \end{equation}
 It follows that the rank of \( \rho^{(N-M)} \) is precisely that of the linear transformation \( \Psi_M \). Furthermore, we have also shown that \( \mathrm{rank}\left[ \rho^{(N-M)}\right]=\mathrm{rank}\left[ \rho^{(M)}\right] \), since, from equation (\ref{eq:psimatrixel}), we can conclude that, as matrices,
 \begin{equation}
 \Psi_{N-M}=\Psi_M^T\,.
 \end{equation}
The interaction length is the longest when the ranks of all the reduced density matrices achieve their maximal possible value. For values of \( M \) from 1 to a maximum to be determined, the reduced density matrix is full-rank in the symmetric space when \( \mathrm{rank}\left[ \rho^{(M)}\right]=M+1 \). On the other hand, since \(\left[ \rho^{(N-M)}\right]=\mathrm{rank}\left[ \rho^{(M)}\right] \) in all cases, we have that \( \rho^{(N-1)} \) has rank at most 2, \( \rho^{(N-2)} \)  at most 3, and so on. Thus, the last possible value of \( M \) for which it is possible for \( \rho^{(M)} \) to be full rank is \( \lfloor{N/2}\rfloor \). Namely,
 \begin{eqnarray}
 \mathrm{rank}\left[ \rho^{(\lfloor{N/2}\rfloor+1)}\right] &=& \mathrm{rank}\left[ \rho^{(\lceil{N/2}\rceil-1)}\right] \nonumber \\ & \leq &\lceil{N/2}\rceil\leq \lfloor{N/2}\rfloor+1\,.
 \end{eqnarray}
 Therefore, with certainty, \( \mathrm{ker}\left[ \rho^{(\lfloor{N/2}\rfloor+1)}\right] \) is not empty, and we conclude that, for all states, there will be parent Hamiltonians with interaction lengths smaller or equal to \( \lfloor{N/2}\rfloor+1 \).
 
Notice that the states whose optimal bond dimension is larger than \( \lfloor{N/2}\rfloor \) will show interaction lengths smaller than the optimal bond dimension. Although we shall not expand on the concept here, it is true that the set of states with optimal bond dimension smaller or equal to  \( \lfloor{N/2}\rfloor+1 \) are dense in the space of states, and we may call it the set of \emph{generic} states.
 
\subsection{Examples of parent Hamiltonians}
In all successive examples, we will have to supplement the Hamiltonians with a two-particle term that implements the restriction to the symmetric space; i.e. the two-particle projector onto the spin 0 sector, the singlet. This is given by 
\begin{equation}
\left(P_0\right)_{i,i+1}= \frac{1}{4}\left(1- \boldsymbol{ \sigma}_i\cdot\boldsymbol{ \sigma}_{i+1}\right)
\end{equation}

\subsubsection{  \(\text{GHZ}_N\) state }
The \( M \)-particle reduced density matrix for the \( | \mathrm{GHZ}_N\rangle \) state is
\begin{equation}
\rho^{(M)}(\mathrm{GHZ}_N)= \frac{1}{2}\left( |0\rangle\langle0|^{\otimes M}+|1\rangle\langle1|^{\otimes M}\right)\,,
\end{equation}
which is always of rank two. The one-particle reduced density matrix is maximally mixed, while the two-particle reduced density matrix has the states \( |10\rangle  \) and \( |01\rangle \) in its kernel. The corresponding projector is
\begin{eqnarray}
\left(P_{S=1,m=0}\right)_{i,i+1}&=& \frac{1}{4}\left(1+ \boldsymbol{ \sigma}_i\cdot\boldsymbol{ \sigma}_{i+1}- 2 \sigma^z_i \sigma^z_{i+1}\right)\\
&=&  \frac{1}{4}\left(1+ \sigma_i^x \sigma^x_{i+1}+ \sigma_i^y \sigma^y_{i+1}-  \sigma_i^z \sigma^z_{i+1}\right).\nonumber
\end{eqnarray}

Therefore, we identify the \( \mathrm{GHZ}_N \) state as a ground state of the Hamiltonian
\begin{equation}
H_{\mathrm{GHZ}_N}=J\sum_{i=1}^N \boldsymbol{ \sigma}_i\cdot \boldsymbol{ \sigma}_{i+1}- J_z \sum_{i=1}^N \sigma_i^z \sigma_{i+1}^z\,,
\end{equation}
with the conditions \( J_z>0 \) and \( J_z>2J \).

By construction, it is easy to see that this Hamiltonian presents degeneracy, since \( |0\rangle^{\otimes N} \) and \( |1\rangle^{\otimes N} \) are both eigenstates with the same energy.

\subsubsection{ \( W_N \) state }
We have
\begin{equation}
\rho^{(M)}\left( W_N\right)=  \frac{M}{N}|W_M\rangle\langle W_M|+ \frac{N-M}{M}|0\rangle\langle0|^{\otimes M}\,.
\end{equation}
For this expression to be valid for \( M=1,2 \), define \( |W_2\rangle= \left(|01\rangle+|10\rangle\right)/ \sqrt{2} \) and \( |W_1\rangle=|1\rangle \). Therefore, a possible parent Hamiltonian is
\begin{equation}
H_W=  \frac{ \alpha}{8}\sum_{i=1}^N \left[-2\sigma_i^z+  \sigma_i^z \sigma_{i+1}^z - \frac{2 \beta}{ \alpha} \boldsymbol{ \sigma}_i\cdot \boldsymbol{ \sigma}_{i+1}\right]
\end{equation}
with positive \(\alpha\) and \(\beta\). Also in this case, the ground state is degenerate.

\subsubsection{ \( X_N \) state}
The reduced density matrix for the \( | \mathrm{X}_N(z)\rangle \) state in the \( M<N-1 \) particles case reads
\begin{widetext}
\begin{equation}
\rho^{(M)}\left[X_N(z)\right]= \frac{1}{1+|z|^{2N-2} N} \left\{|1\rangle\langle1|^{\otimes M}+|z|^{2N-2}\left[ M |W_M\rangle\langle W_M|+\left(N-M\right) |0\rangle\langle0|^{\otimes M}\right]\right\}\,. 
\end{equation}
\end{widetext}

Therefore, for all \( 3\leq M<N-1 \) the rank of the reduced matrix is 3.
The case \( M=N-1 \) is different, and it reads
\begin{widetext}
\begin{eqnarray}
\rho^{(N-1)}\left(X_N(z)\right)&=&\frac{1}{1+|z|^{2N-2} N} \left\{|1\rangle\langle1|^{\otimes (N-1)} + |z|^{2N-2}|0\rangle\langle0|^{\otimes (N-1)}+ |z|^{2N-2} (N-1) |W_{N-1}\rangle\langle W_{N-1}| \right. \nonumber \\
&&+\quad\left. \sqrt{N}\left( z^{N-1} |0\rangle\langle 1|^{\otimes (N-1)}+ \bar{z}^{N-1} |1\rangle\langle 0|^{\otimes (N-1)}\right)\right\}\,.
\end{eqnarray}
\end{widetext}
The case \( M=2 \) can be included in the general formula for arbitrary \( M \) with the convention that
\begin{equation}\label{eq:w2}
|W_2\rangle= \frac{1}{ \sqrt{2}}\left(|01\rangle+|10\rangle\right)\,.
\end{equation}
Again, this means that \( \rho^{(2)}\left(X_N\right)\) is of rank three, which is actually the size of the symmetric space for two qubits. Thus, the only element of the kernel of the two-particle density matrix is the projector out of the symmetric space, i.e. the projector onto the singlet. In order to have a local Hamiltonian that discriminates this state, we need to go to the three-particle density matrix. The Dicke state with two excitations is the only subspace of the symmetric space that lies in the kernel of \( \rho^{(3)}\left(X_N\right) \), and provides us with the local Hamiltonian  term \( h^{(3)}=|D_2^{(3)}\rangle\langle D_2^{(3)}|\), of interaction length 3.

In order to compute an alternative form for this local Hamiltonian term, we look at the following operator:
\[\prod_{l=0,l\neq k}^n\left(\sum_{i=1}^n \sigma_i^z-n+2l\right)\,.\]
In the symmetric space, it is zero except when acting on \( |D_k^{(n)}\rangle=0  \), on which it has the value \( \prod_{l=0,l\neq k}^n\left(2l-2k\right) \). Thus, we have the projector
\begin{equation}
P^{(n)}_k= \prod_{l=0,l\neq k}^n\frac{\sum_{i=1}^n \sigma_i^z-n+2l}{2(l-k)}\,,
\end{equation}
which, when restricted to the symmetric space, provides us with the relevant vector. On explicit computation for the case at hand, \( |D_2^{(3)}\rangle \), we find the relevant operator
\begin{equation*}
P^{(3)}_2= \frac{1}{8}\left( 3- \sigma_1^z \sigma_2^z- \sigma_1^z \sigma_3^z- \sigma_2^z \sigma_3^z+ 3 \sigma_1^z \sigma_2^z \sigma_3^z- \sum_{i=1}^3\sigma_i^z\right)
\end{equation*}
which provides us with the parent Hamiltonian presented in the main text.

\end{document}